\documentclass[USenglish]{amsart}

\usepackage[utf8]{inputenc}
\usepackage{microtype}
\usepackage{bm,amsfonts,amsmath,amssymb,mathrsfs,bbm,amsthm} 
\usepackage{graphicx, color,hyperref,eurosym, eucal,palatino}

\newcommand\sfS{\mathsf{S}}

\usepackage{hyperref}
\usepackage{enumitem}
\usepackage{xcolor, amssymb}
\usepackage{marginnote}
\usepackage[normalem]{ulem}
\usepackage{eurosym}
\usepackage{stackrel}
\definecolor{darkred}{rgb}{0.65,0,0}

\numberwithin{figure}{section}
\numberwithin{equation}{section}
\theoremstyle{plain}
\newtheorem*{thm*}{Theorem}
\newtheorem{thm}{Theorem}[section]

\newtheorem{prop}[thm]{Proposition}

\theoremstyle{remark}
\newtheorem{remark}[thm]{Remark}
\theoremstyle{plain}

\theoremstyle{definition}

\newcommand\cyl{{\rm cyl}}

\newcommand\omg{{\omega}}

\newcounter{counter_a}
\newenvironment{myenum}{\begin{list}{{\rm(\roman{counter_a})}}%
		{\usecounter{counter_a}
			\setlength{\itemsep}{1.ex}\setlength{\topsep}{0.8ex}
			\setlength{\leftmargin}{5ex}\setlength{\labelwidth}{5ex}}}{\end{list}}

\newcommand\wt{\widetilde}

\newcommand\lm{\lambda}
\newcommand\p{\partial}

\def\sfW{{\mathsf W}}
\def\sfM{{\mathsf M}}
\def\sfJ{{\mathsf J}}
\def\sfV{{\mathsf V}}

\newcommand{\bx}{x}

\newcommand{\be}{\mathbf{e}}
\newcommand{\bd}{\mathbf{d}}

\newcommand\ii{{\mathsf{i}}}
\newcommand\dl{\delta}
\newcommand\s{\sigma}
\newcommand\tm{\times}

\newcommand{\dC}{\mathbb{C}}
\newcommand{\dR}{\mathbb{R}}

\newcommand{\dT}{\mathbb{T}}

\newcommand{\dI}{\mathbb{I}}
\newcommand{\dZ}{\mathbb{Z}}
\newcommand{\dN}{\mathbb{N}}
\newcommand{\sfD}{\mathsf{D}}
\newcommand\arr{\rightarrow}

\newcommand\dd{{\mathsf{d}}}
\newcommand{\spn}{\mathsf{span}\,}

\def\cF{{\mathcal F}}

\def\cE{{\mathcal E}}

\def\cD{{\mathcal D}}

\def\cS{{\mathcal S}}

\newcommand\dom{\mathrm{dom}\,}

\renewcommand\tt{\theta}

\newcommand\Op{\sfD_\omega}
\newcommand\wOp{{\wt\sfD}_\omega}

\title[Self-adjoint extensions of the two-valley Dirac operator]{Self-adjoint extensions of the two-valley Dirac operator
with {discontinuous} infinite mass boundary conditions}

\author{Biagio Cassano}
\address{Department of Mathematics, %
  Universit\`{a} degli Studi di  Bari,%
  via Edoardo Orabona 4, 70125, Bari, Italy}
\email{biagio.cassano@uniba.it}

\author{Vladimir Lotoreichik}
\address{Department of Theoretical Physics, Nuclear Physics Institute, 	Czech Academy of Sciences, 25068 \v Re\v z, Czech Republic}
\email{lotoreichik@ujf.cas.cz}
\urladdr{http://gemma.ujf.cas.cz/~lotoreichik}


\begin{document}

\begin{abstract}
  We consider the four-component two-valley  Dirac operator 
  on a wedge in $\dR^2$
  with infinite mass boundary conditions, which enjoy a flip at the
  vertex. We show that it has deficiency indices $(1,1)$ and we parametrize
  all its self-adjoint 
  extensions, relying on the fact that the underlying
  two-component Dirac operator is symmetric with deficiency indices $(0,1)$. The respective
  defect element is computed explicitly. 

  We observe that there exists no self-adjoint extension, which
  can be decomposed into an orthogonal sum of two two-component operators. In physics, 
  this effect is called \emph{mixing the valleys}.
\end{abstract}			
		\maketitle
\section{Introduction}

The dynamics of low-energy electrons in graphene is effectively described by a Hamiltonian associated to the matrix
differential expression 
\[
	\mathcal{M}
	= \begin{pmatrix}\cD & 0\\
	0 & \cD
	\end{pmatrix},
\]
where $\cD$ is the two-component Dirac
differential expression in two dimensions. 
Such a Hamiltonian takes into account contributions from the
two inequivalent Dirac points (or valleys) of the first Brillouin
zone associated to the underlying hexagonal lattice. The respective components of a wavefunction describe the electronic density
on each of the two triangular sublattices that constitute the
honeycomb lattice.
In order to define rigorously the operator
associated to $\mathcal{M}$, appropriate boundary
conditions have to be imposed, and its domain of self-adjointness has
to be determined. In many applications the two valleys are decoupled and the description is reduced to the study of an operator associated to $\cD$ only. However, interactions that mix the valleys may indeed occur in graphene~\cite{TTT06}
and the effects produced by them are often
appearing under the name \emph{valleytronics}; see~\cite{NGPNG09}
and the references therein.
In this paper we consider a discontinuous \emph{infinite mass}
boundary condition and, in order to get self-adjointness for the operator
associated to
$\mathcal{M}$, it is necessary to couple the two valleys.

Following our program, we  investigate the  two-dimensional massless Dirac operator
with {discontinuous} infinite mass boundary conditions on a wedge in the situation when the boundary condition undergoes a flip at the vertex.
This problem can be regarded as a counterpart of the analysis in~\cite{LO18, PV19} 
for a similar problem without a flip.
Following the strategy of \cite{LO18}, in order to obtain the main result we rely on
separation of the variables
and subsequent careful analysis of the one-dimensional fiber operators.
We would like to emphasize that the observed effect is essentially not
caused by the corner of the wedge,
because it persists even if the flip happens on the half-plane.
In this respect it is reminiscent of a similar effect for the Robin
Laplacian with the coefficient having a linear singularity at a boundary point~\cite{ES88, MR09, NP18}.
We expect that relying on the localisation technique given in~\cite{NP18},
our results
can be generalized for 
operators on smooth planar domains
and even on curvilinear polygons,
having (finitely many) flips of the boundary condition. The literature on Dirac operators
with infinite mass boundary conditions on domains {is} quite extensive; see {\it e.g.} \cite{ALTR17, AMV14, BEHL18, BHOP19, HOBP18, OBV18}, the review papers~\cite{BEHL,OP19}, and the references therein.

To describe our main result we need to introduce some notations. In what follows, we consider a wedge: 
\begin{equation}\label{key}
	\cS_\omg
	:= 
	\big\{(r\cos\tt,r\sin\tt)\in\dR^2\colon 
	r >0,
	\tt \in \dI_\omg\big\}
	\subset \dR^2,
\end{equation}
where $\dI_\omg := (-\omg,\omg)$ with $\omg \in (0,\pi)$. 
The value $2\omg$ can be viewed as the opening angle of the wedge $\cS_\omg$.
The opposite sides of the wedge $\cS_\omg$
are denoted by
\[
	\Gamma_\omg^\pm 
		:= 
	\{(r\cos\omg,\pm r\sin\omg)\in\dR^2\colon 
	r >0\}.
\]
Clearly, the choice $\omg = \frac{\pi}{2}$
corresponds to the half-plane.

Recall that the $2\tm 2$ Hermitian 
\emph{Pauli matrices} $\s_1,\s_2,\s_3$ are given by
\[
\s_1 =  \begin{pmatrix}0&1\\1&0\end{pmatrix},
\qquad
\s_2  = \begin{pmatrix}0&-\ii\\ \ii&0\end{pmatrix}
\quad \text{and} \quad 
\s_3 = \begin{pmatrix}1&0\\0&-1\end{pmatrix}. 
\]
For $i,j\in\{1,2,3\}$, they satisfy the anti-commutation relation
$\s_j\s_i + \s_i\s_j = 2 \dl_{ij}$,
where $\dl_{ij}$ is the Kronecker symbol. 
For the sake of convenience, we define 
$\s := (\s_1,\s_2)$ and for 
$\bx = (x_1, x_2)^\top\in\dR^2$ we set
\[
	\s\cdot \bx 
	:= 
	x_1 \s_1 + x_2\s_2 
	= 
	\begin{pmatrix} 
	0 & x_1 -\ii x_2\\
	x_1 + \ii x_2 & 0
	\end{pmatrix}.
\]
Consider the following matrix differential expression
\[
	\cD := -\ii (\s\cdot\nabla) =
	\begin{pmatrix} 
	0                    & -\ii(\p_1 - \ii\p_2)
	\\
	-\ii(\p_1 + \ii\p_2) & 0
	\end{pmatrix}. 
\]
%
%
The subject of our analysis is the Dirac operator $\Op$ in the Hilbert
space $L^2(\cS_\omg;\dC^2)$, defined as follows:
\begin{equation}\label{eq:Op}
\begin{aligned}
	\Op u & := \cD u,\\
	\dom\Op & := 
	\left\{
		u = \begin{pmatrix}u_1\\u_2\end{pmatrix}\in H^1(\cS_\omg;\dC^2)\colon 
		\begin{matrix}
		u_2|_{\Gamma_\omg^+} 
		=	
		- e^{+\ii\omg}u_1|_{\Gamma_\omg^+}\\
		u_2|_{\Gamma_\omg^-} 
		=	
		- e^{-\ii\omg}u_1|_{\Gamma_\omg^-}
		\end{matrix}
	\right\}.
\end{aligned}
\end{equation}
Denoting $\mathbf{n}:=\mathbf{n}(x)$ the outer unit normal at
the point $x \in \partial \cS_\omg \setminus \{0\} =
\Gamma_\omg^- \cup \Gamma_\omg^+$, an
explicit computation shows that the boundary conditions in
\eqref{eq:Op}
are equivalent to 
\begin{equation}\label{eq:infinite.mass}
  u =  \mp \ii \sigma_3 (\sigma \cdot \mathbf{n})
  \,
  u,
  \quad \text{ on } \Gamma_\omg^\pm.
\end{equation}
We remark that the standard realization of the Dirac operator on a wedge with
infinite mass boundary conditions prescribes that
\begin{equation*}
  u =   - \ii \sigma_3 (\sigma \cdot \mathbf{n}) \, u,
  \quad \text{ on } \partial\cS_\omg,
\end{equation*}
while in \eqref{eq:infinite.mass} there is a flip between the boundary conditions
imposed on the opposite sides $\Gamma_\omg^\pm$ of the wedge.
Equivalently, in order to get standard infinite mass boundary conditions one should replace
the second condition
$u_2|_{\Gamma_\omg^-}
=	
- e^{-\ii\omg}u_1|_{\Gamma_\omg^-}$ in~\eqref{eq:Op}
 by $u_2|_{\Gamma_\omg^-} 
=	
e^{-\ii\omg}u_1|_{\Gamma_\omg^-}$.

We show in Proposition~\ref{prop:symmetric}
that the operator $\Op$ is symmetric.
Our first main result concerns the deficiency
indices and subspaces of $\Op$.
\begin{thm}\label{thm1}
	Let the symmetric operator $\Op$ be as in~\eqref{eq:Op}. Then the following properties hold.
	\begin{myenum}
		\item $\Op$ has deficiency indices $(0,1)$.\footnote{For $\sfS\subset\sfS^*$ we 
			adopt the convention
			$n_+(\sfS) := \dim\ker(\sfS^* - \ii)$
			and $n_-(\sfS)  
			:= \dim\ker(\sfS^* + \ii)$.
	The deficiency indices of $\sfS$
	are given by $(n_+(\sfS),n_-(\sfS))$.	
	}
		\item $\ker(\Op^*+ \ii) = \spn\{u_\star\}$ and 
		the defect element is given in polar coordinates by
		\begin{equation}\label{eq:defect}
		u_\star(r,\tt) = \frac{1}{2\sqrt{\omg}}
		\frac{e^{-r}}{\sqrt{r}}
		\begin{pmatrix}
		e^{-\frac{\ii\tt}{2}}
		\\
		-
		e^{\frac{\ii\tt}{2}}	
		\end{pmatrix}.
		\end{equation}
	\end{myenum}	
\end{thm}	
In order to prove Theorem~\ref{thm1}, we
take the advantage of the reformulation
in polar coordinates:
we decompose the operator $\Op$ into an orthogonal sum of infinitely many one-dimensional self-adjoint Dirac operators on the half-line and a momentum-type operator on the half-line, which has deficiency indices $(0,1)$ {and whose defect element
can be explicitly computed by solving an elementary first-order ODE.} 

The full four-component two-valley Dirac operator on a planar domain with infinite mass boundary conditions can be viewed as an orthogonal sum of two two-component (one-valley) Dirac operators
with infinite mass boundary conditions,
in which the unit normals are chosen to point
outwards and inwards, respectively. 
As previously mentioned, the analysis
reduces to the one-valley two-component
Dirac operator unless there is an additional ``off-diagonal" interaction, which
mixes the valleys. 

In our setting,
the two-component Dirac operator
associated with the first valley is precisely given by $\Op$, while the one associated with the second valley 
\[
	\left\{
	u = (u_1,u_2)^\top\in H^1(\cS_\omg;\dC^2)\colon 
	\begin{matrix}
	u_2|_{\Gamma_\omg^\pm} 
	=	
	e^{\pm\ii\omg}u_1|_{\Gamma_\omg^\pm}\\
	\end{matrix}
	\right\}\ni u\mapsto \cD u
\]
is
unitarily equivalent to $-\Op$ via
the Pauli matrix $\s_3$. Hence, 	
the two-valley Dirac operator
is unitarily equivalent to
\begin{equation}\label{eq:Op2valley}
	\sfM_\omg := \Op\oplus(-\Op).
\end{equation}
Clearly, the operator $\sfM_\omg$
is symmetric in $L^2(\cS_\omg;\dC^4)$.

Our second main result concerns
the characterisation of the self-adjoint
extensions for $\sfM_\omg$.
In our model,
mixing the valleys naturally enters
as a necessity to define a self-adjoint Hamiltonian through the coupling constant $\alpha {\in\dT}$, which parametrizes the extension. 
Moreover, this mixing {is inevitable}, since there is no self-adjoint extension of $\sfM_\omg$, which can be represented as an orthogonal sum of two Hamiltonians with respect to the decomposition $	L^2(\cS_\omg;\dC^4) = L^2(\cS_\omg;\dC^2)\oplus L^2(\cS_\omg;\dC^2)$. This mathematical observation still awaits a
thorough physical interpretation.
\begin{thm}\label{thm2}
	Let the symmetric operator $\Op$ be as in~\eqref{eq:Op} and let $u_\star$ be as
	in~\eqref{eq:defect}.
	Then the two-valley Dirac operator
	$\sfM_\omg = \sfD_\omg \oplus (-\sfD_\omg)$
	has deficiency indices $(1,1)$
	and all its self-adjoint extensions
	are given by 
	\[
	\begin{aligned}
	\sfM_{\alpha,\omg}& := 
	\begin{pmatrix}
	\cD u_1 + \ii u_\star\\
	-\cD u_2 - \ii \alpha u_\star
	\end{pmatrix},\\ 
	\dom\sfM_{\alpha,\omg}& := 
	\left\{
	\begin{pmatrix}
	u_1\\ u_2\end{pmatrix}
	+ 
	\begin{pmatrix} 
	u_\star \\ 
	\alpha u_\star
	\end{pmatrix}
	\colon
	u_1,u_2\in \dom\Op
	\right\},
	\end{aligned}
	\]
	where $\alpha\in\dT := \{z\in\dC\colon |z| = 1\}$ is an extension parameter.
\end{thm}	
The proof of Theorem~\ref{thm2} rests upon
Theorem~\ref{thm1} and
classical von Neumann extensions theory; {\it cf}~\cite[{\S X.1}]{reedsimon2}.

\begin{remark}
It is not yet clear if there is a way to single out a distinguished
self-adjoint extension of $\sfM_\omg$.
In this respect, the analysis of the case
without a flip is different:
the two-dimensional Dirac operator is essentially self-adjoint whenever $0<\omg\leq \pi/2$
and for $\pi/2 < \omg < \pi$ it has a unique extension such that its domain
is included in $H^{\frac12}(\cS_\omg;\dC^2)$; {\it cf.}~\cite{LO18} for the
infinite mass boundary condition and~\cite{PV19} for more general
quantum-dot boundary conditions.
In our case, Theorem~\ref{thm2} shows that the regularity of the operator
domain can not be a criterion for selection,
because it is impossible to single out an extension requiring that its domain is included in a Sobolev space
$H^s(\cS_\omg;\dC^4)$, for some specific $s>0$. Indeed, in our setting for any $0<\omg<\pi$ all the extensions have a function in the domain that has a
singularity $\sim |x|^{-\frac12}$ {at} the origin. An analogous
phenomenon was observed in
\cite[{Rem. 1.10}]{CP2018} and \cite[{Rem. 1.11}]{CP2019}
{for Dirac operators with critical Coulomb-type
	spherically symmetric perturbations.}
\end{remark}
 
\subsection*{Organisation of the paper}
We prove in Section~\ref{sec:Radial.Op} that the operator $\Op$ is symmetric
and obtain its equivalent representation
in polar coordinates.
Then, we decompose
the operator $\Op$ into orthogonal sum of 
one-dimensional fiber operators in
Section~\ref{sec:orthogonal}.
Finally, Theorems~\ref{thm1} and~\ref{thm2} are proven in Section~\ref{sec:proof}.

\section{Preliminary analysis of $\Op$}\label{sec:Radial.Op}

\subsection{Symmetry}
In order to prove symmetry of $\Op$ we employ
integration by parts. 
Thanks to the specific
choice of the boundary condition, the boundary
term vanishes.

We denote
by $(\cdot,\cdot)_{\cS_\omg}$ the inner product in $L^2(\cS_\omg;\dC^2)$.
Note that all {the} inner products in the present paper
are linear in the first entry.
\begin{prop}\label{prop:symmetric}
	The operator $\Op$ is densely defined
	and symmetric in the Hilbert space $L^2(\cS_\omg;\dC^2)$.
\end{prop}	
\begin{proof}
  The operator is densely defined in $L^2(\cS_\omg;\dC^2)$,
  because $C^\infty_0(\cS_\omg;\dC^2)\subset\dom\Op$ is dense in $L^2(\cS_\omg;\dC^2)$.
  Since $\cS_\omg$ is the epigraph of a globally Lipschitz function, it is straightforward to derive
  from~\cite[Thm. 3.34 and 3.38]{McL} that the Green's identity
  \[
  		\int_{\cS_\omg} (-\ii\sigma\cdot\nabla) u \cdot \overline{v} \, \dd x
  		-
  		\int_{\cS_\omg}  u \cdot \overline{(-\ii\sigma\cdot\nabla)v} \, \dd x
  		\\
  		=  -\ii \int_{\Gamma_\omg^+\cup\Gamma_\omg^-} ((\sigma\cdot \mathbf{n})
  		\, u) \cdot \overline{v} \, \dd s	
  \]
  holds for all $u,v\in H^1(\cS_\omg;\dC^2)$; {\it cf.}~\cite[Lem. 1.4\,(i)]{PV19} for the same formula on bounded piecewise-$C^1$ domains.
  Hence, for any $u,v\in \dom\Op$ we have that
  \begin{equation}\label{eq:symmetry}
      (\Op u,v)_{\cS_\omg} - (u,\Op v)_{\cS_\omg}
      = -\ii \int_{\Gamma_\omg^+\cup\Gamma_\omg^-} ((\sigma\cdot \mathbf{n})
      \, u) \cdot \overline{v} \, \dd s\\
  \end{equation}
  Thanks to
  the boundary conditions \eqref{eq:infinite.mass},
  we have that
  \begin{equation*}
    \begin{split}
    \int_{\Gamma_\omg^+\cup\Gamma_\omg^-} ((\sigma\cdot \mathbf{n})
    \, u) \cdot \overline{v} \, \dd s
    = &
    \int_{\Gamma_\omg^+} (\sigma \cdot \mathbf{n})
    (-\ii \sigma_3 (\sigma \cdot \mathbf{n}))
    u \cdot \overline{(-\ii \sigma_3 (\sigma \cdot \mathbf{n}))v}
    \, \dd s
    \\
    & +
    \int_{\Gamma_\omg^-} (\sigma \cdot \mathbf{n})
    (\ii \sigma_3 (\sigma \cdot \mathbf{n}))
    u \cdot \overline{(\ii \sigma_3 (\sigma \cdot \mathbf{n}))v}
    \, \dd s.
  \end{split}
\end{equation*}
Since $\pm \ii \sigma_3 (\sigma\cdot \mathbf{n})$ are symmetric
$\dC^{2\times 2}$ matrices, we have that
  \begin{equation*}
    \begin{split}
    \int_{\Gamma_\omg^+\cup\Gamma_\omg^-} ((\sigma\cdot \mathbf{n})
    \, u) \cdot \overline{v} \, \dd s
    = &
    -\int_{\Gamma_\omg^+\cup \Gamma_\omg^+} \sigma_3 (\sigma \cdot
    \mathbf{n})
    (\sigma \cdot \mathbf{n})
    \sigma_3 (\sigma \cdot \mathbf{n})
    u \cdot \overline{v}
    \, \dd s
    \\
    = &
    - \int_{\Gamma_\omg^+\cup\Gamma_\omg^-} ((\sigma\cdot \mathbf{n})
    \, u) \cdot \overline{v} \, \dd s,
     \end{split}
   \end{equation*}
   where in the last equality we have used the fact that $(\sigma\cdot\mathbf{n})^2 = \sigma_3^2 =
   \mathbb{I}_2$.
   We conclude that the right hand side in \eqref{eq:symmetry} vanishes, and
   consequently that $\Op$ is symmetric.
\end{proof}

\subsection{Representation in polar coordinates}\label{sec:Op} 
Let us  introduce polar coordinates $(r,\tt)$
on $\cS_\omg$. They are related to the Cartesian
coordinates $\bx = (x_1,x_2)$ \emph{via} the identities
\[
	\bx(r,\tt) = 
	\begin{pmatrix}
	x_1(r,\tt)\\
	x_2(r,\tt)
	\end{pmatrix},\quad\text{where}\quad x_1 = x_1(r,\tt) = r\cos\tt,\quad
	x_2 = x_2(r,\tt) = r\sin\tt,
\]
for all $r > 0$ and $\tt\in \dI_\omg = (-\omg,\omg)$.
Further, we consider the moving frame 
$(\be_{\rm rad}, \be_{\rm ang})$ associated with the polar coordinates
\[
	\be_{\rm rad}(\tt) 
	=
	\frac{\dd \bx}{\dd r}
	= 
	\begin{pmatrix}
		\cos\tt\\
		\sin\tt
	\end{pmatrix}
	\quad\text{and}\quad
	\be_{\rm ang}(\tt) 
	= 
	\frac{\dd \be_{\rm rad}}{\dd\tt}
	 = 	
	\begin{pmatrix}
	-\sin\tt\\
	\cos\tt
	\end{pmatrix}.
\]
The Hilbert space $L^2_\cyl(\cS_\omg;\dC^2) := L^2(\dR\times \dI_\omg,\dC^2;r\dd r \dd \theta)$ can be viewed
as the tensor product $L^2_r(\dR_+)\otimes L^2(\dI_\omg;\dC^2)$,
where the weighted $L^2$-space $L^2_r(\dR_+)$ is defined as
\[
	L^2_r(\dR_+) = \left\{\psi\colon\dR_+\arr\dC\colon \int_{\dR_+}|\psi|^2r\dd r < \infty\right\}. 
\]
Let us consider the unitary transform
\[
	\sfV \colon L^2(\cS_\omg;\dC^2) \arr L^2_\cyl(\cS_\omg;\dC^2),\qquad 
	(\sfV v)(r,\theta) = u\big(r\cos\theta,r\sin\theta\big),
\]
and introduce the cylindrical Sobolev space by
\[
	H^1_\cyl(\cS_\omg;\dC^2) 
	:= \sfV\big(H^1(\cS_\omg;\dC^2)\big) 
	 = 
	{\Big\{v \colon v,\p_r v, r^{-1}(\p_\theta v) \in L^2_\cyl(\cS_\omg;\dC^2)\Big\}}.
\]
We consider the operator acting in the Hilbert space $L^2_\cyl(\cS_\omg;\dC^2)$
{and} defined as
\begin{equation}\label{eqn:unitequivD}
	\wOp := \sfV \Op \sfV^{-1},
	\qquad 
	\dom\wOp := \sfV\big(\dom\Op\big).
\end{equation}
Now, let us compute the action of $\wOp$ on a function $v\in \dom\wOp$. First, notice that there exists {a unique} $u\in \dom\Op$ such that $v = \sfV u$ and the partial derivatives of $u$ with respect to the Cartesian variables $(x_1,x_2)$
can be expressed through those of $v$ with respect to  polar variables $(r,\tt)$  via the standard relations (for $\bx = \bx(r,\tt)$)
\[
\begin{aligned}
	(\p_1 u)(\bx) & = 
	\cos\tt (\p_r v)(r,\tt) - \sin\tt\frac{(\p_\tt v)(r,\tt)}{r},\\
	(\p_2 u)(\bx) &=
	\sin\tt (\p_r v)(r,\tt) + \cos\tt\frac{(\p_\tt v)(r,\tt)}{r}.
\end{aligned}
\]
Using the latter formul{\ae} we can express the action
of the differential expression $\cD = -\ii(\s\cdot\nabla)$ in polar coordinates as follows
(for $\bx = \bx(r,\tt)$)
\[
\begin{aligned}
	(\cD u)(x) 
	 = 
	-\ii\begin{pmatrix} 
		\p_1 u_2(\bx) - \ii \p_2 u_2(\bx)\\
		\p_1 u_1(\bx) + \ii \p_2 u_1(\bx)
	\end{pmatrix}
	 =
	-\ii\begin{pmatrix} 
		e^{-\ii\tt} (\p_r v_2)(r,\tt) -
		\ii e^{-\ii\tt}r^{-1}(\p_\tt v_2)(r,\tt) 
		\\
		e^{\ii\tt} (\p_r v_1)(r,\tt) 
		+ \ii e^{\ii\tt}r^{-1}(\p_\tt v_1)(r,\tt)
	\end{pmatrix}.
\end{aligned}                        
\]
Note that a basic computation yields
\begin{equation}\label{eq:matrix_identity}
	\s\cdot \be_{\rm rad} 
	= 
	\cos\tt\s_1 + \sin\tt\s_2 
	= 
	\begin{pmatrix}
		0           & e^{-\ii\tt}\\
		e^{\ii\tt}  &0
	\end{pmatrix}.
\end{equation}
Hence, the operator
$\wOp$ acts as
\begin{equation}\label{eqn:expopcyl}
\begin{aligned}
	\wOp v & = 
	-\ii(\s\cdot \be_{\rm rad})
	\left(\p_r v + \frac{v}{2r}- \frac{(-\ii\s_3\p_\tt +\frac12) v}{r}\right),\\
	\dom\wOp & = 
	\big\{ v\in H_{\rm cyl}^1(\cS_\omg;\dC^2)\colon v_2(\cdot,\pm\omg) = -e^{\pm\ii\omg} v_1(\cdot,\pm\omg)\big\}.
\end{aligned}
\end{equation}
\section{Orthogonal decomposition}\label{sec:orthogonal}
Now, we introduce an auxiliary
spin-orbit-type operator in the Hilbert space $(L^2(\dI_\omg;\dC^2), (\cdot,\cdot)_{\dI_\omg})$ 
as follows
\begin{equation}\label{eq:K}
\begin{aligned}
	\sfJ_\omg\phi 
	& = 
	-\ii\s_3\phi' + \frac{\phi}{2} =
	\begin{pmatrix}
	-\ii\phi_1' + \frac{\phi_1}{2}\\
	+\ii\phi_2' + \frac{\phi_2}{2}
	\end{pmatrix},\\[0.4ex]
	\dom\sfJ_\omg 
	& = 
	\Big\{\phi = (\phi_1,\phi_2)\in H^1(\dI_\omg;\dC^2)
	\colon \phi_2(\pm\omg) 
	= 
	-e^{\pm\ii\omg} \phi_1(\pm\omg)\Big\}.
\end{aligned}
\end{equation}
Let us investigate the spectral properties of $\sfJ_\omg$.
\begin{prop}\label{prop:K} 
	Let the operator $\sfJ_\omg$ be as in~\eqref{eq:K}.
	Then the following hold.
	\begin{myenum}
		\item\label{itm:K1} $\sfJ_\omg$ is self-adjoint and has a compact resolvent.
		\item\label{itm:K2} 
		$\s(\sfJ_\omg) = \left\{\lm_k\right\}_{k\in\dZ} = \left\{\frac{\pi k}{2\omg}\right\}_{k\in\dZ}$ 
		and $\cF_k := \ker\big(\sfJ_\omg - \lm_k\big) = 
		\spn\{\phi_k\}$, where 
		\begin{equation}\label{eq:defn.phi_k}
			\phi_k (\theta)=
				\frac{1}{2\sqrt{\omg}}
		\begin{pmatrix}
		e^{+\ii(\lm_k - \frac12)\tt}
		\\
		(-1)^{k+1}
		e^{-\ii(\lm_k - \frac12)\tt}	
              \end{pmatrix}
              ;
      \end{equation}
      moreover, $\{\phi_k\}_{k \in\dZ}$ is 
      {an} orthonormal basis of $L^2(\dI_\omg;\dC^2)$.
		\item $(\s\cdot \be_{\rm rad}) \phi_k =
		(-1)^{k+1} \phi_{-k}$
		{for all $k\in\dZ$.}
\end{myenum}
\end{prop}
\begin{proof}
	\noindent (i)
	The operator $\sfJ_\omg - \frac12$ can be viewed as a momentum operator on a graph with two edges of length $2\omg$, in which the vectors $\phi^{\rm out} := \{\phi_1(-\omg), \phi_2(\omg) \}$
	and $\phi^{\rm in} := \{\phi_1(\omg),
	\phi_2(-\omg)\}$ are connected
	as $\phi^{\rm out} = U\phi^{\rm in}$ via the unitary matrix
	\[
	U = 
	\begin{pmatrix} 
	0 & -e^{\ii\omg}\\
	-e^{\ii\omg} & 0
	\end{pmatrix}.
	\] 
	Hence, 
	$\sfJ_\omg - \frac12$ is self-adjoint 
	by~\cite[Prop. 4.1]{E12} and has a compact
	resolvent by~\cite[Thm. 5.1]{E12}.
	Adding a constant $\frac12$ has no impact on these properties and hence the claim follows.

	\medskip
	
	\noindent (ii)
	Let $\phi = (\phi_1,\phi_2)^\top\in\dom\sfJ_\omg$ and $\lambda\in\dR$ be such that $\sfJ_\omg \phi =\lm \phi$.
	The eigenvalue equation on $\phi$ reads as follows
	\[
	\begin{aligned}
		-\ii\phi_1' +\frac{\phi_1}{2}  = \lm\phi_1,\\
		+\ii\phi_2' +\frac{\phi_2}{2}  = \lm\phi_2.
	\end{aligned}
	\]
	The generic solution of the above system of differential equations is given by
	\[
		\begin{cases}
		\phi_1(\tt) 
		= a_1 
		e^{+\ii(\lambda - \frac12) \tt},\\
		\phi_2(\tt) 
		= a_2 
		e^{-\ii(\lambda - \frac12) \tt},
		\end{cases} 
		\qquad a_1,a_2\in\dC.
	\]	
	Hence, the boundary conditions
	yield
	\[
		\begin{cases}
		a_1 e^{+\ii\omg}e^{+\ii(\lambda - \frac12) \omg} + a_2e^{-\ii(\lambda - \frac12) \omg} = 0,\\
		a_1 e^{-\ii\omg}e^{-\ii(\lambda - \frac12) \omg} + a_2e^{+\ii(\lambda - \frac12) \omg} = 0,
		\end{cases}
	\]
	that can be simplified as
	\[
	\begin{cases}
		a_1 e^{+\ii(\lambda + \frac12) \omg} + a_2e^{-\ii(\lambda - \frac12) \omg} = 0,\\
		a_1e^{-\ii(\lambda + \frac12) \omg} + a_2e^{+\ii(\lambda - \frac12) \omg} = 0.
	\end{cases}
	\]
	This system has a non-trivial solution if the corresponding
        determinant vanishes, that is
	\[
		\Delta = 
		 e^{+\ii(\lambda + \frac12) \omg}
		 e^{+\ii(\lambda - \frac12) \omg}
		 -
		 e^{-\ii(\lambda - \frac12) \omg}
		 e^{-\ii(\lambda + \frac12) \omg}
		 = e^{+2\ii\lm\omg} - e^{-2\ii\lm\omg}
		 = 2\ii\sin(2\lm\omg),
	\]
	and consequently the eigenvalues are given by
	\[
		\lm_k = \frac{\pi k}{2\omg},\qquad k\in\dZ.
	\]
	The corresponding eigenvectors can be recovered with the aid of the formula
	\[
		a_1e^{+\ii\left(\frac{\pi k}{2} + \frac{\omg}{2}\right)} + a_2
		e^{-\ii(\frac{\pi k}{2} - \frac{\omg}{2})} = 0
	\]
	which leads to	$a_1 e^{\ii\pi k} + a_2 = 0$. The choice
	\[
		a_1 = \frac{1}{2\sqrt{\omg}},\qquad
		a_2 = \frac{(-1)^{k+1}}{2\sqrt{\omg}}
	\]
	yields the orthonormal basis in \eqref{eq:defn.phi_k}.
	\medskip
	
	\noindent (iii) 
	Using~\eqref{eq:matrix_identity} we obtain
	\[
	\begin{aligned}
			(\s\cdot \be_{\rm rad}) \phi_k & =
		\frac{1}{2\sqrt{\omg}}
			\begin{pmatrix}
			0           & e^{-\ii\tt}\\
			e^{\ii\tt}  &0
			\end{pmatrix} 
				\begin{pmatrix}
			e^{+\ii(\lm_k - \frac12)\tt}
			\\
			(-1)^{k+1}
			e^{-\ii(\lm_k - \frac12)\tt}	
			\end{pmatrix}\\
			& =
		\frac{1}{2\sqrt{\omg}}
			\begin{pmatrix}
			(-1)^{k+1} e^{-\ii(\lm_k
				+\frac12)\tt}\\
			e^{\ii(\lm_k + \frac12)\tt}
			\end{pmatrix}\\
			& =
			\frac{1}{2\sqrt{\omg}}(-1)^{k+1}
			\begin{pmatrix}
			 e^{\ii(\lm_{-k}
				-\frac12)\tt}\\
			(-1)^{k+1}e^{-\ii(\lm_{-k} - \frac12)\tt}
			\end{pmatrix} = (-1)^{k+1}\phi_{-k}.
			\qedhere
	\end{aligned}
	\]
\end{proof}
Further, we employ the orthogonal decomposition
\[
	L^2_\cyl(\cS_\omg; \dC^2\big)\simeq 
	L^2_r(\dR_+)\otimes L^2(\dI_\omg;\dC^2) 
	= 
	\oplus_{k\in\dN_0} \cE_k,
\]
where 
$\cE_0 = L^2_r(\dR)\otimes\cF_0$ and
$\cE_k = L^2_r(\dR)\otimes (\cF_k\oplus\cF_{-k})$ for $k \in\dN$.
In the following proposition we show that
$\cE_k$ are reducing subspaces for $\wOp$.
The analysis of $\wOp$ {boils} down to the study of its restrictions to these subspaces.
For the sake of convenience, we introduce the unitary transforms $\sfW_0 \colon \cE_0
\arr L^2(\dR_+)$
and
$\sfW_k \colon \cE_k
\arr L^2(\dR_+;\dC^2)$ for $k \in\dN$
as
\[
	(\sfW_0 u)(r) :=
	\sqrt{r}\,
	\big(u(r,\cdot),\phi_0\big)_{\dI_\omg},
	\quad
	(\sfW_k u)(r) :=
	\sqrt{r}
	\begin{pmatrix}
	(u(r,\cdot),\phi_k)_{\dI_\omg}\\
	\ii (u(r,\cdot),\phi_{-k})_{\dI_\omg}
	\end{pmatrix}.
\]
\begin{prop}\label{prop:decomp}
 	For any $k\in\dN_0$,   
	\[
		d_k u := \wOp u,
		\qquad
		\dom d_k := \dom\wOp \cap \cE_k,
	\]
	is a well-defined operator in the Hilbert space $\cE_k$. 

	The operator $d_0$ is unitarily equivalent
	via $\sfW_0$ to the operator $\bd_0$ in the Hilbert space $L^2(\dR_+)$ defined as
	\begin{equation}\label{eq:op_dk}
		\bd_0 \psi :=\ii\psi',\qquad 
		\dom\bd_0 :=		H^1_0(\dR_+).
	\end{equation}

	For any $k \in \dN$, the operator $d_k$ is unitarily equivalent via
	$\sfW_k$ to the operator $\bd_k$ in the Hilbert space $L^2(\dR_+)$ defined as
	\begin{equation}\label{eq:op_dk}
		\bd_k :=
		(-1)^{k+1}
		\begin{pmatrix}
			0 	& -\frac{\dd}{\dd r} - \frac{\pi k}{2\omg r}\\
			\frac{\dd}{\dd r} - \frac{\pi k}{2\omg r}	&0
		\end{pmatrix},\qquad
		\dom{\bd_k}  := H^1_0(\dR_+;\dC^2).
	\end{equation}

	In particular, the decomposition
	\[
		\Op \simeq \bigoplus_{k\in\dN_0} \bd_k
	\]
	holds and the deficiency indices
	of $\Op$ can be computed as	$n_\pm(\Op)  = \sum_{k\in\dN_0} n_\pm(\bd_k)$.
\label{prop:fibdecdisk}
\end{prop}
\begin{proof}
\noindent \underline{{\it Step 1:} $k=0$}.
Pick a function $u\in\dom\wOp\cap \cE_0$. By definition, $u$ writes as
\[
	u(r,\tt) = \frac{\psi_0(r)}{\sqrt{r}}\phi_0(\tt),
\]
with some $\psi_0\colon\dR_+\arr\dC$.
Next, we observe that
$u\in H_{\rm cyl}^1(\cS_\omg;\dC^2)$
is equivalent to $u,\p_r u, \frac{\p_\tt u}{r} \in L^2_r(\dR_+)$, which is, in its turn, equivalent to $\psi_0,(\frac{\psi_0}{\sqrt{r}})'\sqrt{r},
\frac{\psi_0}{r}\in L^2(\dR_+)$.
Now, we aim at showing the following equivalence
\begin{equation}\label{eq:equivalence_u_psi}
	u\in H^1_{\rm cyl}(\cS_\omg;\dC^2)
	\quad\Longleftrightarrow\quad\psi_0 \in
		H^1_0(\dR_+).
\end{equation} 	
First, we obtain that
\[
	\psi_0' =
	\left(\frac{\psi_0}{\sqrt{r}}\right)'\sqrt{r} + \frac12\frac{\psi_0}{r} \in L^2(\dR_+). 
\]
Hence, 
$u\in H^1_{\rm cyl}(\cS_\omg;\dC^2)$  implies  $\psi_0 \in
H^1(\dR_+)$.
Moreover, thanks e.g. to~\cite[Prop. 2.2 {\rm (i)}]{CP2018} (with $a=0$ settled there)
we infer that there exists $\mathsf{p}\in\dC$ such that
\begin{equation*}
  \lim_{r\to 0^+} |\psi_0(r) - \mathsf{p}| r^{-1/2}= 0,
\end{equation*}
and, according to ~\cite[Prop. 2.4 {\rm (i)}]{CP2018} (for $a=0$), we 
obtain that $\frac{\psi_0 - \mathsf{p}}{r} \in L^2(\dR_+)$. Since
$\frac{\psi_0}{r}\in L^2(\dR_+)$, 
we get that $\mathsf{p} =0$. Hence,
by the Sobolev trace theorem
we obtain that $\psi_0 \in H^1_0(\dR_+)$. The
reverse implication in~\eqref{eq:equivalence_u_psi} immediately
follows from the one-dimensional Hardy inequality; see e.g. \cite[Prop. 2.4 {\rm (i)}]{CP2018}.

Applying the differential expression obtained in~\eqref{eqn:expopcyl} to $u$, we get
\begin{equation}\label{eq:k=0}
	(\wOp u)(r,\theta) 	
	=
	-\ii(\s\cdot \be_{\rm rad})\phi_0(\tt)
	\left(\p_r\left(
		\frac{\psi_0(r)}{\sqrt{r}}\right) 
		+ 
		\frac{\psi_0(r)}{2r^{3/2}}
	\right)=
	\ii \frac{\psi'_0(r)}{\sqrt{r}}\phi_0(\tt).
\end{equation}

\medskip

\noindent 
\underline{{\it Step 2:} $k\in\dN$.}
Pick a function $u\in\dom\wOp\cap \cE_k$. By definition, $u$ writes as
\[
	u(r,\tt) =
	\frac{\psi_{+k}(r)}{\sqrt{r}}\phi_{k}(\tt)
	-\ii
	\frac{\psi_{-k}(r)}{\sqrt{r}}\phi_{-k}(\tt)
	,
\]
with some $\psi_{\pm k}\colon\dR_+\arr\dC$.
Observe that
\begin{equation}\label{eq:orthophi}
	\begin{aligned}
	(\phi_k',\phi_{-k}')_{\dI_\omg}
	&=
	-((\ii\s_3)^2\phi_k',\phi_{-k}')_{\dI_\omg}
	 =
	(-\ii\s_3\phi_k',-\ii\s_3\phi_{-k}')_{\dI_\omg}\\
	& =
	\left(\lm_k - \frac12\right)\left(\lm_{-k}-\frac12\right)
	(\phi_k,\phi_{-k})_{\dI_\omg}
	=0.
	\end{aligned}
\end{equation}
Again, $u\in H_{\rm cyl}^1(\cS_\omg;\dC^2)$
is equivalent to $u,\p_r u, \frac{\p_\tt u}{r} \in L^2_r(\dR_+)$.
Taking into account orthogonality~\eqref{eq:orthophi},
$u\in H_{\rm cyl}^1(\cS_\omg;\dC^2)$
is equivalent to
$\psi_{\pm k},(\frac{\psi_{\pm k}}{\sqrt{r}})'\sqrt{r},
\frac{\psi_{\pm k}}{r}\in L^2(\dR_+)$
and as in the case $k=0$
we end up with equivalence between $u\in H^1_{\rm cyl}(\cS_\omg;\dC^2)$ and $\psi_{\pm k} \in H^1_0(\dR_+)$.
Applying the differential expression obtained in \eqref{eqn:expopcyl}, we get
\begin{equation}\label{eq:kne0}
	\begin{aligned}
	\wOp u
	&= 
	-
	\frac{\ii(\s\cdot \be_{\rm rad})}{\sqrt{r}}
	\left[\phi_k
	\left(
		\p_r\psi_{k}
		-\frac{\lm_k\psi_k}{r}
	\right) 
	-\ii
	\phi_{-k}
	\left(
	\p_r\psi_{-k}  
		-\frac{\lm_{-k}\psi_{-k}}{r}
	\right)\right]  
\\
	&= 
	\frac{(-1)^{k+1}}{\sqrt{r}}
	\left[
	-\ii\phi_{-k}
	\left(\p_r\psi_{k}  
	-\frac{\lm_k\psi_k}{r}\right)
	+
	\phi_k
	\left(
	-\p_r\psi_{-k}  
	-\frac{\lm_{k}\psi_{-k}}{r}
	\right)\right].
	\end{aligned}
\end{equation}
\noindent \underline{{\it Step 3:} Conclusion of the proof.}
The analysis in Steps 1 and 2
yields that the inclusion $\wOp \left(\dom\wOp\cap\cE_k\right) \subset \cE_k$
holds
for all $k\in\dN_0$. Hence, the operators
$d_k$ are symmetric for all $k\in\dN_0$.
Relying on formulae~\eqref{eq:k=0} and~\eqref{eq:kne0} we find that
\[
	\sfW_k d_k\sfW^{-1}_k = \bd_k,
	\qquad\forall\,k\in\dN_0. \qedhere
\]
\end{proof}

\section{Proofs of the main results}
\label{sec:proof}

With all the preparations above the proofs of the main results are rather
compact.

\begin{proof}[Proof of Theorem~\ref{thm1}]

For all $k\in\dN$ the operators $\bd_k$ are self-adjoint thanks to
\cite[Thm.~1.1\,(i) and Prop.~3.1\,(i)]{CP2018},
since for all $k \in \dN$ we have $\gamma:= \left|\frac{k \pi}{2\omega}\right|>\frac12$.

By a direct computation
it is elementary to observe  that
\[
	\bd_0^* = \ii \psi',\qquad\dom\bd_0^* = H^1(\dR_+).
\]
Hence,
\[
	\ker(\bd_0^* - \ii) = \{0\}
	\quad\text{and}\quad\ker(\bd_0^* + \ii) = \spn\{e^{-r}\}.	
\]
The deficiency indices of $\bd_0$ are given by $(0,1)$ and the corresponding defect element is $\psi_\star(r) = e^{-r}$. 
Hence, by Proposition~\ref{prop:decomp}
the operators $\wOp$ and $\Op$ have deficiency
indices $(0,1)$ as well and the
defect element of $\sfD_\omg$ is given
in polar coordinates by
\[
	u_\star(r,\tt) =
	\big(\sfW_0^{-1}\psi_\star\big)(r,\tt) =
	\frac{e^{-r}}{\sqrt{r}}\phi_0(\tt) =
	\frac{1}{2\sqrt{\omg}}
	\frac{e^{-r}}{\sqrt{r}}
	\begin{pmatrix}
	e^{-\frac{\ii\tt}{2}}
	\\
	-
	e^{\frac{\ii\tt}{2}}	
	\end{pmatrix}.\qedhere
\] 
\end{proof}
\begin{proof}[Proof of Theorem~\ref{thm2}]
Since the operator $\Op$ has deficiency indices
$(0,1)$, the operator $-\Op$ has deficiency
indices $(1,0)$, respectively, and moreover
$\ker(\Op^* +\ii) = \ker((-\Op)^* - \ii) = \spn\{u_\star\}$. 
Therefore, the deficiency indices of the operator $\sfM_\omg = \Op\oplus (-\Op)$ are  $(1,1)$ and its defect subspaces are {given by}
\[
	\ker(\sfM_\omg^* -\ii) = 
	\spn\left\{\begin{pmatrix} 0 \\ u_\star\end{pmatrix} \right\} \qquad\text{and}\qquad
	\ker(\sfM_\omg^* + \ii) = 
	\spn\left\{\begin{pmatrix} u_\star\\ 0\end{pmatrix} \right\}. 
\]
Hence, by~\cite[Thm. X.2]{reedsimon2} all the self-adjoint extensions
of $\sfM_\omg$ are parametrized by
$\alpha\in\dT$ as follows
\[
	\sfM_{\alpha,\omg} := 
	\begin{pmatrix}
		\cD u_1 + \ii u_\star\\
		-\cD u_2 -\ii \alpha u_\star
	\end{pmatrix}, 
	\quad
	\dom\sfM_{\alpha,\omg} := 
	\left\{
		\begin{pmatrix}u_1\\ u_2\end{pmatrix}
		 + 
		\begin{pmatrix} 
		u_\star \\ 
		\alpha u_\star
		\end{pmatrix}
		\colon
		u_1,u_2\in \dom\Op
	\right\},
\]
by which the proof is concluded.
\end{proof}
\subsection*{Acknowledgement}
The authors are grateful for the support by the grant No.~17-01706S of the Czech Science Foundation (GA\v{C}R).

\newcommand{\etalchar}[1]{$^{#1}$}

\end{document}